\documentclass[12pt,oneside,english]{amsart}

\usepackage[T1]{fontenc}
\usepackage[latin9]{inputenc}
\usepackage{geometry}
\usepackage{babel}
\usepackage{amsbsy}
\usepackage{amstext}
\usepackage{amsthm}
\usepackage{amssymb}
\usepackage[unicode=true,pdfusetitle,
 bookmarks=true,bookmarksnumbered=false,bookmarksopen=false,
 breaklinks=false,pdfborder={0 0 1},backref=false,colorlinks=false]
 {hyperref}

\makeatletter
\numberwithin{equation}{section}
\numberwithin{figure}{section}
\theoremstyle{plain}
\newtheorem{thm}{\protect\theoremname}[section]
\theoremstyle{remark}
\newtheorem{rem}[thm]{\protect\remarkname}
\theoremstyle{definition}
\newtheorem{defn}[thm]{\protect\definitionname}
\theoremstyle{plain}
\newtheorem{prop}[thm]{\protect\propositionname}
\theoremstyle{plain}
\newtheorem{lem}[thm]{\protect\lemmaname}
\theoremstyle{definition}
\newtheorem{example}[thm]{\protect\examplename}

\usepackage[all,arc,curve,rotate]{xy}

\makeatother

\providecommand{\definitionname}{Definition}
\providecommand{\examplename}{Example}
\providecommand{\lemmaname}{Lemma}
\providecommand{\propositionname}{Proposition}
\providecommand{\remarkname}{Remark}
\providecommand{\theoremname}{Theorem}

\begin{document}
\title{Transforming antiunitary symmetries to a normal form}
\author{Terry A. Loring}
\address{Department of Mathematics and Statistics, University of New Mexico,
Albuquerque, New Mexico, 87131, USA}
\email{tloring@unm.edu}
\begin{abstract}
We look at explicit ways to bring one or two antiunitary symmetries
into a standard form via unitary conjugation. We carefully reproduce
Wigner's proof in two special cases, where the antiunitary operators
square to $+I$, or to $-I$. Wigner's method is constructive and
we show how leads to two algorithms to compute the needed unitaries
in small examples. 

We then adapt these algorithms to deal with two such antiunitary matrices
that commute up to a sign. This leads to a proof a finite-dimensional
version of the well-known ten-fold way of topological physics. This
will allow physicists to perform a change of basis on any finite-dimensional
model in one of the Altland-Zirnbauer symmetry classes to a standardized
version of that symmetry class in which time-reversal and and particle-hole
symmetry are standard operations that can be implemented efficiently
in standard numerical software. This will also make it easier to find
formulas for some key isomorphisms between graded real $C^{*}$-algebras.
\end{abstract}

\maketitle

\section{Introduction}

When done on Hilbert space, some physics calculations are basis dependent.
A subtle concept like fermionic parity \cite{grabsch2019pfaffian}
depends on which basis is used, but even more, depends of the ordering
of the basis. Another subtle point is that antilinear operators do
not behave in a familiar way when we do a change of basis. When working
by hand, we often can just use the existence of a satisfactory basis.
For example, if an antiunitary $\mathcal{U}$ squares to $-I$ we
can use Kramer's degeneracy to show that there is always an orthogonal
basis that put $\mathcal{U}$ into a standard form. A good basis for
more general antiunitary symmetries are known to exist, by a classic
result of Wigner \cite{wigner1960normal}. This leads to algorithms,
not much more difficult than the Gram-Schmidt algorithm, that can
quickly find such a basis. 

We will use mathematics notation. Most critically, $M^{*}$ denotes
the complex transpose, while $\overline{\text{M}}$ denotes complex
conjugation of each matrix element. We can describe then the transpose
operation, using the formula
\[
M^{\top}\mathbf{v}=\overline{\left(M^{*}\overline{\mathbf{v}}\right)}.
\]

Mathematicians often avoid work directly with antiunitary symmetries.
Instead they look at operations similar to the transpose that operate
on $C^{*}$-algebras. One such can be added to it structure to make
what is called a real $C^{*}$-algebra. Two such, that commute, lead
to a graded, real $C^{*}$-algebra \cite{el-kaioum2014graded}.

One can look at any antiunitary symmetry as adding a real structure
to a $C^{*}$-algebra, here just the $C^{*}$-algebra of all complex
$n$-by-$n$ matrices. Often, working at the more abstract level of
real $C^{*}$-algebras is much easier than working with antilinear
operator. However, there are some explicit isomorphisms that are found
more easily by looking at the underlying antiunitary symmetries.

It is when doing computer calculations that we really due have to
pick a good basis. For example, linear algebra software will have
versions of operations that are optimized for real matrices. If we
work with matrices that would only become real after applying a change
of basis, we lose this optimization. Even worse, numerical errors
can make our matrices have only an approximate reality condition.
Storing ``real matrices'' as actual real matrices is easier, faster,
more accurate and will use less memory. 

Physicists frequently need to consider having one or two order-two
antiunitaries that square to one, up to sign, and that commute, up
to sign. This comes up regarding the single-particle Hamiltonians
that are in specific Altland-Zirnbauer symmetry classes \cite{altland1997nonstandard,kaufmann2015topological}.
In the traditional setting for topological insulators these are $\mathcal{T}$
and $\mathcal{C}$, with the first time-reversal and the second particle-charge
conjugation. In practice, $\mathcal{T}$ can end up as a physical
rotation composed with actual time-reversal. The ten-fold way tells
us that such pairs, for a given choice of the signs, are unique up
to conjugation by a single unitary operator. Here we focus on finding
these unitaries, or equivalently, finding an orthonormal basis on
which these antiunitary operators acts in a simple, canonical manner.

There are many subtle points about when various symmetries need to
unitary or antiunitary. For example, it matters if one is working
in first quantization or second quantization \cite{Zirnbauer2021Particle-Hole-symmetries}.
We are ignoring these issues and focus on how to deal with symmetries
implemented by one or two antiunitary operators. The basic outline
of this paper follows the methods used to classify elementary graded
real $C^{*}$-algebras \cite{el-kaioum2014graded}. 

Recall that an operator $\mathcal{J}:\mathfrak{H}\rightarrow\mathfrak{H}$
is \emph{antilinear} if 
\[
\mathcal{J}(\alpha\mathbf{v}+\beta\mathbf{w})=\overline{\alpha}\mathcal{J}(\mathbf{v})+\overline{\beta}\mathcal{J}(\mathbf{w}).
\]
The cannonical example here is $\mathcal{K}:\mathbb{C}^{n}\rightarrow\mathbb{C}^{n}$.
Since the product of two antilinear operators is linear, we can always
factor $\mathcal{J}$ as $\mathcal{J}=T\circ\mathcal{K}$ with $T$
linear. When $\mathfrak{H}=\mathbb{C}^{n}$ we want to associate an
$n$-by-$n$ matrix $M$ with its associated linear operator $\mathbf{v}\mapsto M\mathbf{v}$.
We use the calligraphic font to denote antilinear operators, and use
the composition symbol. With these conventions, we find 
\[
M\circ\mathcal{K}\circ N\circ\mathcal{K}=M\circ\overline{N}=M\overline{N}.
\]
Note that in the physics literature, the word operator can sometime
indicate an antilinear operator. 

An operator $\mathcal{T}:\mathfrak{H}\rightarrow\mathfrak{H}$ is
\emph{antiunitary} if it is antilinear, onto, and if
\[
\left\langle \mathcal{T}\left(\mathbf{v}\right),\mathcal{T}\left(\mathbf{w}\right)\right\rangle =\overline{\left\langle \mathbf{v},\mathbf{w}\right\rangle }
\]
for all vectors. In the standard factoring $\mathcal{T}=U\circ\mathcal{K}$
we will have $U$ a unitary operator. Wigner, in 1960, proved there
is a unitary $Q$ that conjugates $\mathcal{T}$ to a simpler antiunitary
in a prescribed form. That is, 
\[
Q^{*}\circ\mathcal{T}\circ Q=\mathcal{T}_{0}
\]
where there are orthogonal subspaces left invariant by $\mathcal{T}_{0}$,
and on each subspace the action is a copy of one of the $\mathcal{W}_{\theta}$
we now describe. For a random $\mathcal{T}$, most of the basic blocks
will be a copy of $\mathcal{W}_{\theta}$, the antiunitary that acts
as
\[
\mathcal{W}_{\theta}\left(\left[\begin{array}{c}
\alpha\\
\beta
\end{array}\right]\right)=\left[\begin{array}{c}
e^{i\theta}\overline{\beta}\\
\overline{\alpha}
\end{array}\right].
\]
The restriction on $\theta$ is $0\leq\theta\leq\pi$, with $\theta=0$
being a special case. In this special case we define $\mathcal{W}_{0}$
to act on $\mathbb{C}$ as
\[
\mathcal{W}_{0}\left(\alpha\right)=\overline{\alpha}.
\]
The one-by-one blocks have action that is just complex conjugation.
\begin{rem}
Wigner finds blocks that act like 
\[
\widehat{\mathcal{W}}_{\theta}\left(\left[\begin{array}{c}
\alpha\\
\beta
\end{array}\right]\right)=\left[\begin{array}{c}
e^{\frac{i\theta}{2}}\overline{\beta}\\
e^{-\frac{i\theta}{2}}\overline{\alpha}
\end{array}\right].
\]
The unitary $Q_{\theta}$ defined by
\[
Q_{\theta}\left(\left[\begin{array}{c}
\alpha\\
\beta
\end{array}\right]\right)=\left[\begin{array}{c}
\beta\\
e^{\frac{i}{2}\theta}\alpha
\end{array}\right]
\]
intertwines this with $\mathcal{W}_{\theta}$, as one can check $Q_{\theta}\circ\widehat{\mathcal{W}}_{\theta}=\mathcal{W}_{\theta}\circ Q_{\theta}$.
One can use $\mathcal{W}_{\theta}$ or $\widehat{\mathcal{W}}_{\theta}$
as one prefers.
\end{rem}

The only symmetries allowed in a framework like the ten-fold way are
those that square to $\pm I$. So here we will only find all the $\theta$
are zero or all are $\pi$. 
\begin{defn}
Let $\mathcal{T}$ be an antiunitary operator. The order of $\mathcal{T}$
is the smallest $n$ for which $\mathcal{T}^{n}=\pm I$ for some real
$\theta$. If there is none such, then this is of infinite order.
An antiunitary operator $\mathcal{T}$ of order two is \emph{even}
if $\mathcal{T}\circ\mathcal{T}=I$ and \emph{odd} if $\mathcal{T}\circ\mathcal{T}=-I$.
\end{defn}

\section{The structure of even antiunitary operators}

So now we assume $\mathcal{T}=U\circ\mathcal{K}$ and $\mathcal{T}\circ\mathcal{T}=+I$.
Wigner's theorem will guarantee, in this special case, the existence
of unitary matrix $Q$ so that
\begin{equation}
Q^{*}\circ\mathcal{T}\circ Q=\mathcal{K}.\label{eq:conugation_to-conjugation}
\end{equation}
If we turn this into a problem on the unitary matrix $U$ we find
\begin{align*}
I\circ\mathcal{K} & =Q^{*}\circ\mathcal{T}\circ Q\\
 & =Q^{*}\circ U\circ\mathcal{K}\circ Q\\
 & =Q^{*}\circ U\circ\overline{Q}\circ\mathcal{K}.
\end{align*}
We cancel the $\mathcal{K}$ (since its one-to-one) and find $I=Q^{*}U\overline{Q}$
or 
\[
QQ^{\top}=U.
\]
In some cases one can deduce a value for $Q$ just from looking at
this and the equation $QQ^{*}=I$. A more systematic approach is to
follow Wigner's proof from \cite{wigner1960normal}. 

Wigner's proof simplifies when we assume $\mathcal{T}\circ\mathcal{T}=I$.
We deal with the case $\mathcal{T}\circ\mathcal{T}=-I$ in the next
section. The part of Wigner's proof we follow in this section can
be found roughly in Equations 12 through 18 in \cite{wigner1960normal}.
Let us state this special form of the theorem in terms of an orthogonal
basis. These basis elements will be the columns of $Q$ to get the
conjugation to standard form shown in Equation~\ref{eq:conugation_to-conjugation}.

Proposition~\ref{prop:Even_antiunitary_structure} can be used in
class D as well as class AI. The only difference is that we call the
one symmetry $\mathcal{C}$.
\begin{prop}
\label{prop:Even_antiunitary_structure} Suppose $\mathcal{T}$ is
an antiunitary operator on finite-dimensional Hilbert space such that
\[
\mathcal{T}\circ\mathcal{T}=+I.
\]
Then there is an orthonormal basis $\boldsymbol{\mathbf{\psi}}_{1},\dots,\boldsymbol{\mathbf{\psi}}_{n}$
such that 
\[
\mathcal{T}\left(\boldsymbol{\mathbf{\psi}}_{j}\right)=\boldsymbol{\mathbf{\psi}}_{j}
\]
for all $j$.
\end{prop}

\begin{proof}
Take any unit vector $\boldsymbol{\varphi}_{1}$. Then
\[
\mathcal{T}\left(\boldsymbol{\varphi}_{1}+\mathcal{T}\left(\boldsymbol{\varphi}_{1}\right)\right)=\boldsymbol{\varphi}_{1}+\mathcal{T}\left(\boldsymbol{\varphi}_{1}\right)
\]
since $\mathcal{T}\left(\mathcal{T}\left(\mathbf{v}\right)\right)=\mathbf{v}$
for any vector. We now create our first basis vector $\boldsymbol{\mathbf{\psi}}_{1}$
in one of two ways. If $\mathcal{T}\left(\boldsymbol{\varphi}_{1}\right)\neq-\boldsymbol{\varphi}_{1}$
then we set
\[
\boldsymbol{\psi}_{1}=\frac{\boldsymbol{\varphi}_{1}+\mathcal{T}\left(\boldsymbol{\varphi}_{1}\right)}{\left\Vert \boldsymbol{\varphi}_{1}+\mathcal{T}\left(\boldsymbol{\varphi}_{1}\right)\right\Vert }.
\]
If $\mathcal{T}\left(\boldsymbol{\varphi}_{1}\right)=-\boldsymbol{\varphi}_{1}$
then we set 
\[
\boldsymbol{\psi}_{1}=i\boldsymbol{\varphi}_{1}.
\]
In both cases, we now check that $\mathcal{T}(\boldsymbol{\psi}_{1})=\boldsymbol{\psi}_{1}$.
In the first case we have 
\begin{align*}
\mathcal{T}\left(\frac{\boldsymbol{\varphi}_{1}+\mathcal{T}\left(\boldsymbol{\varphi}_{1}\right)}{\left\Vert \boldsymbol{\varphi}_{1}+\mathcal{T}\left(\boldsymbol{\varphi}_{1}\right)\right\Vert }\right) & =\frac{\mathcal{T}\left(\boldsymbol{\varphi}_{1}+\mathcal{T}\left(\boldsymbol{\varphi}_{1}\right)\right)}{\left\Vert \boldsymbol{\varphi}_{1}+\mathcal{T}\left(\boldsymbol{\varphi}\right)\right\Vert }\\
 & =\frac{\boldsymbol{\varphi}_{1}+\mathcal{T}\left(\boldsymbol{\varphi}_{1}\right)}{\left\Vert \boldsymbol{\varphi}_{1}+\mathcal{T}\left(\boldsymbol{\varphi}_{1}\right)\right\Vert }
\end{align*}
while in the second case (keeping in mind that $\mathcal{T}$ is antilinear)
we have
\[
\mathcal{T}\left(i\boldsymbol{\varphi}_{1}\right)=-i\mathcal{T}\left(\boldsymbol{\boldsymbol{\varphi}}\right)=-i\left(-\boldsymbol{\boldsymbol{\varphi}}\right)=i\boldsymbol{\varphi}_{1}.
\]

Now assume $\boldsymbol{\psi}_{1},\cdots,\boldsymbol{\psi}_{k-1}$
have been selected, with $k$ not more than the dimension of the Hilbert
space. Pick any unit vector $\boldsymbol{\phi}_{k}$ that is orthogonal
to each of these $k-1$ vectors. We now create $\mathbf{\boldsymbol{\psi}}_{k}$
in one of two ways. If $\mathcal{T}\left(\boldsymbol{\varphi}_{k}\right)\neq-\boldsymbol{\varphi}_{k}$
we set 
\[
\boldsymbol{\psi}_{k}=\frac{\boldsymbol{\varphi}_{k}+\mathcal{T}\left(\boldsymbol{\varphi}_{k}\right)}{\left\Vert \boldsymbol{\varphi}_{k}+\mathcal{T}\left(\boldsymbol{\varphi}_{k}\right)\right\Vert }.
\]
If $\mathcal{T}\left(\boldsymbol{\varphi}_{k}\right)=-\boldsymbol{\varphi}_{k}$
we set 
\[
\boldsymbol{\psi}_{k}=i\boldsymbol{\varphi}_{k}.
\]
In all these formulas we retain orthogonality to the $\boldsymbol{\psi}_{j}$
for $j<k$, since vector addition, the antiunitary $\mathcal{T}$,
and scalar multiplication all preserve this orthogonality. Thus $\boldsymbol{\psi}_{k+1}$
will be a unit vector that is orthogonal to the proceeding $\boldsymbol{\psi}_{j}$
so the construction may proceed.
\end{proof}
If $\mathcal{T}=U\circ\mathcal{K}$ for $U$ some unitary operator
we can just state this all in terms of $U$ and the conjugation of
vectors. Just for clarity, this says $\mathcal{T}\left(\boldsymbol{v}\right)=U\overline{\boldsymbol{v}}.$ 

We want at the end to have 
\[
U\overline{\boldsymbol{\psi}_{j}}=\boldsymbol{\psi}_{j}.
\]
From $\boldsymbol{\varphi}_{k}$ we now describe $\boldsymbol{\psi}_{k}$
in terms of $U$ and conjugation. If $U\overline{\boldsymbol{\varphi}_{k}}\neq-\boldsymbol{\varphi}_{k}$
then 
\[
\boldsymbol{\psi}_{k}=\frac{\boldsymbol{\varphi}_{k}+U\overline{\boldsymbol{\varphi}_{k}}}{\left\Vert \boldsymbol{\varphi}_{k}+U\overline{\boldsymbol{\varphi}_{k}}\right\Vert }.
\]
 If $U\overline{\boldsymbol{\varphi}_{k}}\neq\boldsymbol{\varphi}_{k}$
then 
\[
\boldsymbol{\psi}_{k}=i\boldsymbol{\varphi}_{k}.
\]

\section{The structure of a odd antiunitary operators}
\begin{lem}
Suppose $\mathcal{T}$ is an antiunitary operator such that $\mathcal{T}\circ\mathcal{T}=-I$.
If $\boldsymbol{\varphi}$ is any vector then $\mathcal{T}(\boldsymbol{\varphi})$
will be orthogonal to $\boldsymbol{\varphi}$.
\end{lem}

\begin{proof}
Recall the property $\left\langle \mathcal{T}\left(\boldsymbol{\varphi}\right),\mathcal{T}\left(\boldsymbol{\psi}\right)\right\rangle =\left\langle \boldsymbol{\psi},\boldsymbol{\varphi}\right\rangle $
of all antiunitary operators. We now have $\mathcal{T}\left(\mathcal{T}\left(\mathbf{v}\right)\right)=-\mathbf{v}$
and so find
\begin{align*}
\left\langle \boldsymbol{\varphi},\mathcal{T}\left(\boldsymbol{\varphi}\right)\right\rangle  & =\left\langle \mathcal{T}\left(\mathcal{T}\left(\boldsymbol{\varphi}\right)\right),\mathcal{T}\left(\boldsymbol{\varphi}\right)\right\rangle \\
 & =\left\langle -\boldsymbol{\varphi},\mathcal{T}\left(\boldsymbol{\varphi}\right)\right\rangle \\
 & =-\left\langle \boldsymbol{\varphi},\mathcal{T}\left(\boldsymbol{\varphi}\right)\right\rangle 
\end{align*}
and so $\left\langle \boldsymbol{\varphi},\mathcal{T}\left(\boldsymbol{\varphi}\right)\right\rangle =0$. 
\end{proof}
In this special case, Wigner's theorem shows the existence of unitary
matrix $Q$ so that
\begin{equation}
Q\circ\mathcal{T}\circ Q^{*}=\mathcal{T}_{0}\label{eq:conugation_to-conjugation-1}
\end{equation}
where $\mathcal{T}_{0}=U_{0}\circ\mathcal{K}$ and $U_{0}$ is just
copies of $i\sigma_{y}$. That is
\[
U_{0}=\left[\begin{array}{ccccc}
0 & 1\\
-1 & 0\\
 &  & 0 & 1\\
 &  & -1 & 0\\
 &  &  &  & \ddots
\end{array}\right].
\]
We follow Wigner, equations (21) to (24), in the proof of the following.
\begin{prop}
\label{prop:Odd_antiunitary_structure} Suppose $\mathcal{T}$ is
an antiunitary operator on finite-dimensional Hilbert space such that
\[
\mathcal{T}\circ\mathcal{T}=-I.
\]
Then there is an orthonormal basis whose vectors come in pairs $\boldsymbol{\phi}_{j},\boldsymbol{\psi}_{j}$
with the property that 
\[
\mathcal{T}\left(\boldsymbol{\varphi}_{j}\right)=-\boldsymbol{\psi}_{j},\ \mathcal{T}\left(\boldsymbol{\psi}_{j}\right)=\boldsymbol{\varphi}_{j}.
\]
 
\end{prop}

\begin{proof}
We pick $\boldsymbol{\varphi}_{1}$ as any unit vector, then simply
set 
\[
\boldsymbol{\psi}_{1}=-\mathcal{T}\left(\boldsymbol{\varphi}_{1}\right).
\]
This is also a unit vector, and is orthogonal to $\boldsymbol{\varphi}_{1}$
by the above lemma. If we apply $\mathcal{T}$ to this equation we
get
\[
\mathcal{T}\left(\boldsymbol{\psi}_{1}\right)=\mathcal{T}\left(-\mathcal{T}\left(\boldsymbol{\varphi}_{1}\right)\right)=-\mathcal{T}\left(\mathcal{T}\left(\boldsymbol{\varphi}_{1}\right)\right)=\boldsymbol{\varphi}_{1}.
\]

Assume $\boldsymbol{\varphi}_{1},\boldsymbol{\psi}_{1},\cdots,\boldsymbol{\varphi}_{j-1},\boldsymbol{\psi}_{j-1}$
have been selected. We select $\boldsymbol{\varphi}_{j}$ as any unit
vector orthogonal to all the previously defined basis elements. We
set 
\[
\boldsymbol{\psi}_{j}=\mathcal{T}\left(\boldsymbol{\varphi}_{j}\right).
\]
We need to check that $\boldsymbol{\psi}_{j}$ is orthogonal to $\boldsymbol{\varphi}_{k}$
and $\boldsymbol{\psi}_{k}$ for $k<j$. We find
\[
\left\langle \boldsymbol{\psi}_{j},\boldsymbol{\varphi}_{k}\right\rangle =\left\langle \mathcal{T}\left(\boldsymbol{\varphi}_{k}\right),\mathcal{T}\left(\boldsymbol{\psi}_{j}\right)\right\rangle =\left(-\boldsymbol{\psi}_{k},\boldsymbol{\varphi}_{j}\right)=0
\]
and 
\[
\left\langle \boldsymbol{\psi}_{j},\boldsymbol{\psi}_{k}\right\rangle =\left\langle \mathcal{T}\left(\boldsymbol{\psi}_{k}\right),\mathcal{T}\left(\boldsymbol{\psi}_{j}\right)\right\rangle =\left\langle \boldsymbol{\varphi}_{k},\boldsymbol{\varphi}_{j}\right\rangle =0.
\]
\end{proof}
If we want to work in terms of $U$ and $\mathcal{K}$, where $\mathcal{T}\mathbf{v}=U\overline{\mathbf{v}}$,
we see this is very simple process. We pick the $\boldsymbol{\varphi}_{j}$
by selecting a unit vector orthogonal to what has been already selected,
then just set $\boldsymbol{\psi}_{j}=-U\overline{\boldsymbol{\varphi}_{j}}$. 

\section{Time-reversal (anti)invariant observables}

Consider a matrix observable $M$ that is Hermitian and that commutes
or anticommutes with $\mathcal{T}=U\circ\mathcal{K}$. For now, we
put no restrictions on $\mathcal{T}$ beyond being anti-unitary. Our
assumption is thus
\[
\mathcal{T}\circ M=\pm M\circ\mathcal{T}.
\]
Some equivalent ways to state this are
\begin{align*}
\mathcal{T}\circ M\circ\mathcal{T}^{-1} & =\pm M
\end{align*}
and 
\begin{align*}
U\circ\mathcal{K}\circ M\circ\mathcal{K}\circ U^{*} & =\pm M
\end{align*}
and finally
\[
U\overline{M}U^{*}=\pm M.
\]
This is equivalent to 
\[
UM^{\top}U^{*}=\pm M
\]
since we have assumed $M^{*}=M$. 

First the case of an odd antiunitary symmetry.
\begin{prop}
Suppose $\mathcal{T}\circ\mathcal{T}=I$ for some antiunitary $\mathcal{T}$
and that $Q$ is a unitary such that $Q^{*}\circ\mathcal{T}\circ Q=\mathcal{K}$.
If $M$ is a Hermitian matrix such that 
\[
\mathcal{T}\circ M=\pm M\circ\mathcal{T}
\]
then $M_{Q}=Q^{*}MQ$ is a matrix such that 
\[
\mathcal{K}\circ M_{Q}=\pm M_{Q}\circ\mathcal{K}.
\]
This conclusion can be restated so that the conclusion is that either
$M_{Q}$ is real symmetric, in the commuting case, or purely-imaginary
and anti-symmetric in the anticommuting case.
\end{prop}

\begin{proof}
The equation $Q^{*}\circ\mathcal{T}\circ Q=\mathcal{K}$ can be rewritten
as either $Q^{*}\circ\mathcal{T}=\mathcal{K}\circ Q^{*}$ or $\mathcal{T}\circ Q=Q\circ\mathcal{K}$.
Using these equalities we find
\begin{align*}
\mathcal{K}\circ M_{Q} & =\mathcal{K}\circ Q^{*}\circ M\circ Q\\
 & =Q^{*}\circ\mathcal{T}\circ M\circ Q\\
 & =Q^{*}\circ\left(\pm M\circ\mathcal{T}\right)\circ Q\\
 & =\pm Q^{*}\circ M\circ\mathcal{T}\circ Q\\
 & =\pm Q^{*}\circ M\circ Q\circ\mathcal{K}\\
 & =\pm M_{Q}\circ\mathcal{K}.
\end{align*}
\end{proof}
Now the odd case.
\begin{prop}
Suppose $\mathcal{T}\circ\mathcal{T}=-I$ and $\mathcal{T}_{0}\circ\mathcal{T}_{0}=-I$
and for some antiunitaries $\mathcal{T}$ and $\mathcal{T}_{0}$.
Further suppose that $Q$ is a unitary such that $Q^{*}\circ\mathcal{T}\circ Q=\mathcal{T}_{0}$.
If $M$ is a Hermitian matrix such that 
\[
\mathcal{T}\circ M=\pm M\circ\mathcal{T}
\]
then $M_{Q}=Q^{*}MQ$ is a matrix such that 
\[
\mathcal{T}_{0}\circ M_{Q}=\pm M_{Q}\circ\mathcal{T}_{0}.
\]
\end{prop}

\begin{proof}
The equation $Q\circ\mathcal{T}\circ Q^{*}=\mathcal{T}_{0}$ can be
rewritten as either $Q\circ\mathcal{T}=\mathcal{T}_{0}\circ Q$ or
$\mathcal{T}\circ Q^{*}=Q^{*}\circ\mathcal{T}_{0}$. Using these equalities
we find
\begin{align*}
\mathcal{T}_{0}\circ M_{Q} & =\mathcal{T}_{0}\circ Q^{*}\circ M\circ Q\\
 & =Q^{*}\circ\mathcal{T}\circ M\circ Q\\
 & =Q^{*}\circ\left(\pm M\circ\mathcal{T}\right)\circ Q\\
 & =\pm Q^{*}\circ M\circ\mathcal{T}\circ Q\\
 & =\pm Q^{*}\circ M\circ Q\circ\mathcal{T}_{0}\\
 & =\pm M_{Q}\circ\mathcal{T}_{0}.
\end{align*}
\end{proof}

\section{An antiunitary symmetry and a grading that commute}

In some Altland-Zirnbauer symmetry classes, we have a grading as well
as an antiunitary $\mathcal{T}:\mathfrak{H}\rightarrow\mathfrak{H}$
that commutes with the grading operator. The grading operator is a
unitary $\Pi$ with $\Pi^{2}=I$. If we start with the particle-charge
conjugation antiunitary $\mathcal{C}$ then the grading operator is
the unitary $\Pi=\mathcal{T}\circ\mathcal{\mathcal{C}}=\mathcal{\mathcal{C}}\circ\mathcal{T}$.
We need not assume that $\Pi$ is be balanced. Here balanced would
mean that the two eigenspaces
\begin{equation}
\mathfrak{H}_{+}=\left\{ \left.\boldsymbol{\varphi}\in\mathfrak{H}\right|\strut\Pi\boldsymbol{\varphi}=\boldsymbol{\varphi}\right\} \label{eq:eigenspace_for_plus_one}
\end{equation}
and
\begin{equation}
\mathfrak{H}_{-}=\left\{ \left.\boldsymbol{\varphi}\in\mathfrak{H}\right|\strut\Pi\boldsymbol{\varphi}=-\boldsymbol{\varphi}\right\} \label{eq:eigenspace_for_neg_one}
\end{equation}
would have to be of the same dimension.
\begin{lem}
Suppose $\mathcal{T}$ is an antisymmetry, and that $\Pi$ is a unitary
with $\Pi^{2}=I$. Moreover, assume $\mathcal{T}\circ\Pi=\Pi\circ\mathcal{T}$.
If $\Pi\mathbf{v}=\pm\mathbf{v}$ then $\Pi\mathcal{T}\left(\mathbf{v}\right)=\pm\mathcal{T}\left(\mathbf{v}\right)$.
\end{lem}

\begin{proof}
We see 
\[
\Pi\mathcal{T}\left(\mathbf{v}\right)=\mathcal{T}\left(\Pi\mathbf{v}\right)=\mathcal{T}\left(\pm\mathbf{v}\right)=\pm\mathcal{T}\left(\mathbf{v}\right).
\]
\end{proof}
This means that in the setting where $\mathcal{T}\circ\Pi=\Pi\circ\mathcal{T}$
and $\mathcal{T}\circ\mathcal{T}=\pm I$ and $\Pi^{2}=I$, we can
restrict these to the eigenspaces of $\Pi$ and get two smaller antiunitary
maps 
\[
\mathcal{T}_{\pm}:\mathfrak{H}_{\pm}\rightarrow\mathfrak{H}_{\pm}.
\]
Both of these will square to $\pm I$, with the same sign as what
holds for $\mathcal{T}$. Thus we can just run the Wigner algorithms
in each sector separately. 

Proposition~\ref{prop:both_even} will apply to $\mathcal{T}$ and
$\mathcal{C}$ in class BDI.
\begin{prop}
\label{prop:both_even} Suppose $\mathcal{T}$ is an antiunitary operator
and $\Pi$ is a unitary on finite-dimensional Hilbert space, with
\[
\mathcal{T}\circ\mathcal{T}=+I,\ \Pi^{2}=I,\ \Pi\circ\mathcal{T}=\mathcal{T}\circ\Pi.
\]
Then there is an orthonormal basis $\boldsymbol{\psi}_{1},\dots,\boldsymbol{\psi}_{n}$
where every basis vector is invariant under $\mathcal{T}$ and for
some $m$ we have $\Pi\boldsymbol{\psi}_{j}=\boldsymbol{\psi}_{j}$
for $1\leq j\leq m$ and $\Pi\boldsymbol{\psi}_{j}=-\boldsymbol{\psi}_{j}$
for $m<j\leq n$.
\end{prop}

\begin{proof}
The proof is to just apply the Wigner's even construction to $\mathfrak{H}_{+}$
and again to $\mathfrak{H}_{-}$.
\end{proof}
Given $\mathcal{T}$ and $\Pi$, both squaring to $+I$, Proposition~\ref{prop:both_even}
means we can find a unitary $Q$ with 
\[
Q\circ\mathcal{T}\circ Q^{*}=\mathcal{K}
\]
and also 
\[
Q\Pi Q^{*}=\Pi_{0}
\]
where 
\[
\Pi_{0}=\left[\begin{array}{cc}
I_{m} & 0\\
0 & -I_{n-m}
\end{array}\right].
\]
If $\Pi$ is balanced, you might prefer to define $\Pi_{0}$ so it
is just copies of $\sigma_{z}$ down the diagonal. The key is to realize
you need to modify the grading and time-reversal operator at the same
time. 

Now a result for class CII. 
\begin{prop}
\label{prop:both_odd} uppose $\mathcal{T}$ is an antiunitary operator
and $\Pi$ is a unitary on finite-dimensional Hilbert space, with
\[
\mathcal{T}\circ\mathcal{T}=-I,\ \Pi^{2}=I,\ \Pi\circ\mathcal{T}=\mathcal{T}\circ\Pi.
\]
Then there is some $m$ and an orthonormal basis whose vectors come
in pairs $\boldsymbol{\phi}_{j},\boldsymbol{\psi}_{j}$ with the property
that $\mathcal{T}\left(\boldsymbol{\phi}_{j}\right)=-\boldsymbol{\psi}_{j}$
and $\mathcal{T}\left(\boldsymbol{\psi}_{j}\right)=\boldsymbol{\phi}_{j}$
and $\Pi\left(\boldsymbol{\phi}_{j}\right)=\boldsymbol{\phi}_{j}$
and $\Pi\left(\boldsymbol{\psi}_{j}\right)=\boldsymbol{\psi}_{j}$
for $j\leq m$ and $\Pi\left(\boldsymbol{\phi}_{j}\right)=-\boldsymbol{\phi}_{j}$
and $\Pi\left(\boldsymbol{\psi}_{j}\right)=-\boldsymbol{\psi}_{j}$
for $j>m$. 
\end{prop}

\section{An atiunitary symmetry and a grading that anticommute}

In classes DIII and CI we have two antiunitary symmetries, one odd,
one even, and they anticommute. We just look at DIII as CI can be
dealt with by simply swapping $\mathcal{T}$and $\mathcal{C}$.

We have thus $\mathcal{T}$, an antiunitary operator on finite-dimensional
Hilbert space with $\mathcal{T}\circ\mathcal{T}=-I$, plus a unitary
$\Pi$ with $\Pi^{2}=1$ and $\mathcal{T}\circ\Pi=-\Pi\circ\mathcal{T}$.
We can derive $\mathcal{C}$ from the equation $\mathcal{C}=\mathcal{T}\circ\Pi$
if we need this.
\begin{lem}
Suppose $\mathcal{T}$ is an antiunitary on finite-dimensional Hilbert
space such that $\mathcal{T}\circ\mathcal{T}=-I$ and that $\Pi$
is a unitary with $\Pi^{2}=I$. Moreover, assume $\mathcal{T}\circ\Pi=-\Pi\circ\mathcal{T}$.
If $\Pi\mathbf{v}=\pm\mathbf{v}$ then $\Pi\mathcal{T}\left(\mathbf{v}\right)=\mp\mathcal{T}\left(\mathbf{v}\right)$.
\end{lem}

\begin{proof}
We see 
\[
\Pi\mathcal{T}\left(\mathbf{v}\right)=-\mathcal{T}\left(\Pi\mathbf{v}\right)=-\mathcal{T}\left(\pm\mathbf{v}\right)=\mp\mathcal{T}\left(\mathbf{v}\right).
\]
\end{proof}
A consequence of this is that $\Pi$ is balanced.
\begin{prop}
\label{prop:opposite_signs} Suppose $\mathcal{T}$ is an antiunitary
and $\Pi$ is a unitary on finite-dimensional Hilbert space. Suppose
also that
\[
\mathcal{T}\circ\mathcal{T}=-I,\ \Pi^{2}=I,\ \mathcal{T}\circ\Pi=-\Pi\circ\mathcal{T}.
\]
Then there is an orthonormal basis whose vectors come in pairs $\boldsymbol{\varphi}_{j},\boldsymbol{\psi}_{j}$
with 
\[
\mathcal{T}\left(\boldsymbol{\varphi}_{j}\right)=-\boldsymbol{\psi}_{j},\ \mathcal{T}\left(\boldsymbol{\psi}_{j}\right)=\boldsymbol{\varphi}_{j},\ \Pi\boldsymbol{\varphi}_{j}=\boldsymbol{\varphi}_{j},\ \Pi\boldsymbol{\psi}_{j}=-\boldsymbol{\psi}_{j}.
\]
\end{prop}

\begin{proof}
Here we just follow the proof as in the proof of Proposition~\ref{prop:Odd_antiunitary_structure},
taking care to always select $\mathbf{\phi}_{j}$ in the $+1$ eigenspace
for $\Pi$. 
\end{proof}

\section{Examples}

\subsection{One even antiunitary}
\begin{rem}
The obvious algorithm derived from the proof of \ref{prop:Even_antiunitary_structure}
can be applied to single particle models in both Class AI and Class
D. The only difference in class D is one applies all the above to
$\mathcal{Q}$ instead of $\mathcal{T}$. It is also useful looking
at a derived antiunitary symmetry applying to a spectral localizer
in class AII, dimension two. The random looking unitary \cite[Eqn 6.2]{hastings2010almost}
that adjust that matrix to be actually skew-symmetric can be found
as in the following example. 
\end{rem}

\begin{example}
Suppose $\mathcal{T}=U\circ\mathcal{K}$ where
\[
U=(i\sigma_{y})\otimes(i\sigma_{y})=\left[\begin{array}{cccc}
0 & 0 & 0 & 1\\
0 & 0 & -1 & 0\\
0 & -1 & 0 & 0\\
1 & 0 & 0 & 0
\end{array}\right].
\]
This satisfies $\mathcal{T}\circ\mathcal{T}=+I$ so we follow the
proof of Proposition~\ref{prop:Even_antiunitary_structure}. Let's
pick 
\[
\boldsymbol{\varphi}_{1}=\left[\begin{array}{c}
1\\
0\\
0\\
0
\end{array}\right].
\]
This means
\[
\boldsymbol{\psi}_{1}=\frac{1}{\sqrt{2}}\left(\boldsymbol{\varphi}_{1}+U\overline{\boldsymbol{\varphi}_{1}}\right)=\frac{1}{\sqrt{2}}\left[\begin{array}{c}
1\\
0\\
0\\
1
\end{array}\right].
\]
 Orthogonal to this is 
\[
\boldsymbol{\varphi}_{2}=\left[\begin{array}{c}
0\\
1\\
0\\
0
\end{array}\right]
\]
which leads to 
\[
\boldsymbol{\psi}_{2}=\frac{1}{\sqrt{2}}\left(\boldsymbol{\varphi}_{2}+U\overline{\boldsymbol{\varphi}_{2}}\right)=\frac{1}{\sqrt{2}}\left[\begin{array}{c}
0\\
1\\
-1\\
0
\end{array}\right].
\]
Orthogonal to $\boldsymbol{\psi}_{1}$ and $\boldsymbol{\psi}_{2}$
is 
\[
\boldsymbol{\varphi}_{3}=\frac{1}{\sqrt{2}}\left[\begin{array}{c}
0\\
1\\
1\\
0
\end{array}\right]
\]
which leads to $U\boldsymbol{\varphi}_{3}^{*}=-\boldsymbol{\varphi}_{3}$
and so
\[
\boldsymbol{\psi}_{3}=i\boldsymbol{\varphi}_{3}=\frac{1}{\sqrt{2}}\left[\begin{array}{c}
0\\
i\\
i\\
0
\end{array}\right].
\]
Finally we can select 
\[
\boldsymbol{\varphi}_{4}=\frac{1}{\sqrt{2}}\left[\begin{array}{c}
1\\
0\\
0\\
=1
\end{array}\right]
\]
which leads to $U\overline{\boldsymbol{\varphi}_{4}}=-\boldsymbol{\varphi}_{4}$
and so
\[
\boldsymbol{\psi}_{4}=\left[\begin{array}{c}
i\\
0\\
0\\
-i
\end{array}\right].
\]
The final result is
\[
Q=\frac{1}{\sqrt{2}}\left[\begin{array}{cccc}
1 & 0 & 0 & i\\
0 & 1 & i & 0\\
0 & -1 & i & 0\\
1 & 0 & 0 & -i
\end{array}\right].
\]
\end{example}

Here $Q$ is far from unique. In particular, we might end up with
$\det Q=1$ or $\det Q=-1$. When computing Pfaffians, this will be
mildly annoying. We can use a fixed $Q$ to change basis on everything
(states and observables) and compute in a standard picture. This ambiguity
also shows up in the well-known fact regarding the algebra quaternions
$\mathbb{H}$. The real algebras $\mathbb{H}\otimes\mathbb{H}$ and
$\boldsymbol{M}_{4}(\mathbb{R})$ are isomorphic but there really
isn't a canonical choice of isomorphism. However, $\mathbb{H}$ is
isomorphic to a real algebra of complex matrices \cite{loring2012factorization_quaternions},
\[
\mathbb{H}\cong\left\{ \left.M\in\boldsymbol{M}_{2}(\mathbb{C})\strut\right|M\circ\mathcal{T}_{-}=\mathcal{T}_{-}\circ M\right\} 
\]
for $\mathcal{T}_{-}=i\sigma_{y}\circ\mathcal{K}$. Thus is is obvious
that
\[
\mathbb{H}\otimes\mathbb{H}\cong\left\{ \left.M\in\boldsymbol{M}_{4}(\mathbb{C})\strut\right|M\circ\mathcal{T}=\mathcal{T}\circ M\right\} 
\]
for $\mathcal{T}=\left(i\sigma_{y}\otimes i\sigma_{y}\right)\circ\mathcal{K}$.
The unitary $Q$ in this example gives us one explicit isomorphism,
by conjugation, between these real $C^{*}$-algebras.

See \cite{doll2021skew_localizer}, where many similar odd-looking
unitaries are used to convert antiunitary symmetries that natuarally
apply to the spectral localizer over into standard symmetries as a
necessary step in defining various $\mathbb{Z}_{2}$-valued local
topological invariants. 

\subsection{One odd antiunitary}
\begin{rem}
The algorithm suggested by the proof of Proposition~\ref{prop:Odd_antiunitary_structure}
will apply to $\mathcal{C}$ class C and to $\mathcal{T}$ in class
AII. 
\end{rem}

\begin{example}
Suppose $\mathcal{T}\left(\boldsymbol{v}\right)=U\overline{\boldsymbol{v}}$
where
\[
U=\frac{1}{\sqrt{2}}\left[\begin{array}{cccc}
0 & 0 & 1 & 1\\
0 & 0 & 1 & -1\\
-1 & -1 & 0 & 0\\
-1 & 1 & 0 & 0
\end{array}\right].
\]
This has $\mathcal{T}\circ\mathcal{T}=I$ so we follow the proof of
Proposition~\ref{prop:Odd_antiunitary_structure}. We select 
\[
\boldsymbol{\varphi}_{1}=\left[\begin{array}{c}
1\\
0\\
0\\
0
\end{array}\right].
\]
Then 
\[
\boldsymbol{\psi}_{1}=-U\overline{\boldsymbol{\varphi}_{1}}=-U\boldsymbol{\varphi}_{1}=\left[\begin{array}{c}
0\\
0\\
\tfrac{1}{\sqrt{2}}\\
\tfrac{1}{\sqrt{2}}
\end{array}\right].
\]
We now pick
\[
\boldsymbol{\varphi}_{2}=\left[\begin{array}{c}
0\\
1\\
0\\
0
\end{array}\right]
\]
and get 
\[
\boldsymbol{\psi}_{2}=-U\overline{\boldsymbol{\varphi}_{2}}=-U\boldsymbol{\varphi}_{2}=\left[\begin{array}{c}
0\\
0\\
\tfrac{1}{\sqrt{2}}\\
\tfrac{-1}{\sqrt{2}}
\end{array}\right].
\]
The final result is
\[
Q=\left[\begin{array}{cccc}
1 & 0 & 0 & 0\\
0 & 0 & 1 & 0\\
0 & \tfrac{1}{\sqrt{2}} & 0 & \tfrac{1}{\sqrt{2}}\\
0 & \tfrac{1}{\sqrt{2}} & 0 & \tfrac{-1}{\sqrt{2}}
\end{array}\right].
\]
\end{example}

\subsection{An antiunitary anticommuting with time reversal}

The following example is based on how time-reversal and particle-charge
conjugation work in a physical model in class DIII \cite{diez2014bimodal}.
\begin{example}
Suppose $\mathcal{T}\left(\boldsymbol{v}\right)=U\boldsymbol{v}^{*}$
where
\[
U=I\otimes\sigma_{y}=\left[\begin{array}{cccc}
0 & 0 & -i & 0\\
0 & 0 & 0 & -i\\
i & 0 & 0 & 0\\
0 & i & 0 & 0
\end{array}\right]
\]
and
\[
\Pi=\sigma_{x}\otimes\sigma_{y}=\left[\begin{array}{cccc}
0 & 0 & 0 & -i\\
0 & 0 & -i & 0\\
0 & i & 0 & 0\\
i & 0 & 0 & 0
\end{array}\right].
\]
We find
\[
\mathcal{T}\circ\mathcal{T}=U\circ\mathcal{K}\circ U\circ\mathcal{K}=UU^{*}=-U^{2}=-I
\]
and clearly $\Pi^{2}=1$ and $\Pi U=U\Pi$. Also 
\[
\mathcal{T}\circ\Pi=U\circ\mathcal{K}\circ\Pi=U\circ\Pi^{*}\circ\mathcal{K}=-U\circ\Pi\circ\mathcal{K}=-\Pi\circ U\circ\mathcal{K}=-\Pi\circ\mathcal{T}.
\]
Thus this example sits in class DIII. We will follow the proof of
Proposition~\ref{prop:opposite_signs}.

We can find a single unitary that conjugates these over $\mathcal{T}_{0}$
and $\Pi_{0}$ where $\mathcal{T}_{0}\left(\boldsymbol{v}\right)=U_{0}\overline{\boldsymbol{v}}$
and 
\[
U_{0}=\left[\begin{array}{cccc}
0 & 1 & 0 & 0\\
-1 & 0 & 0 & 0\\
0 & 0 & 0 & 1\\
0 & 0 & -1 & 0
\end{array}\right],
\]
\[
\Pi_{0}=\left[\begin{array}{cccc}
1 & 0 & 0 & 0\\
0 & -1 & 0 & 0\\
0 & 0 & 1 & 0\\
0 & 0 & 0 & -1
\end{array}\right].
\]

We need some eigenvectors for the eigenvalue $1$ for $\Pi$. These
look like
\[
\left[\begin{array}{c}
\alpha\\
\beta\\
i\beta\\
i\alpha
\end{array}\right].
\]
We pick 
\[
\boldsymbol{\varphi}_{1}=\frac{1}{\sqrt{2}}\left[\begin{array}{c}
1\\
0\\
0\\
i
\end{array}\right]
\]
and this gives us 
\[
\boldsymbol{\psi}_{1}=-U\overline{\boldsymbol{\varphi}_{1}}=\frac{1}{\sqrt{2}}\left[\begin{array}{cccc}
0 & 0 & i & 0\\
0 & 0 & 0 & i\\
-i & 0 & 0 & 0\\
0 & -i & 0 & 0
\end{array}\right]\left[\begin{array}{c}
1\\
0\\
0\\
-i
\end{array}\right]=\frac{1}{\sqrt{2}}\left[\begin{array}{c}
0\\
1\\
-i\\
0
\end{array}\right].
\]
Next we pick 
\[
\boldsymbol{\varphi}_{2}=\frac{1}{\sqrt{2}}\left[\begin{array}{c}
0\\
1\\
i\\
0
\end{array}\right]
\]
and this gives us 
\[
\psi_{2}=-U\overline{\boldsymbol{\varphi}_{2}}=\frac{1}{\sqrt{2}}\left[\begin{array}{cccc}
0 & 0 & i & 0\\
0 & 0 & 0 & i\\
-i & 0 & 0 & 0\\
0 & -i & 0 & 0
\end{array}\right]\left[\begin{array}{c}
0\\
-1\\
i\\
0
\end{array}\right]=\frac{1}{\sqrt{2}}\left[\begin{array}{c}
-1\\
0\\
0\\
-i
\end{array}\right].
\]
The needed unitary is
\[
Q=\frac{1}{\sqrt{2}}\left[\begin{array}{cccc}
1 & 0 & 0 & -1\\
0 & 1 & 1 & 0\\
0 & -i & i & 0\\
i & 0 & 0 & -i
\end{array}\right].
\]
To verify, we compute $Q^{*}\Pi Q=\Pi_{0}$ and $Q^{*}UQ^{*}=U_{0}$. 
\end{example}

\section*{Akknowledgements}

Research by T.A.L.~was sponsored by the Army Research Office and
was accomplished under Grant Number W911NF-25-1-0052. The views and
conclusions contained in this document are those of the authors and
should not be interpreted as representing the official policies, either
expressed or implied, of the Army Research Office or the U.S.\textasciitilde Government.
The U.S.~Government is authorized to reproduce and distribute reprints
for Government purposes notwithstanding any copyright notation herein.

\end{document}